\documentclass[11pt]{article}
\usepackage[utf8]{inputenc}
\usepackage{style}
\usepackage{thm-restate}
\usepackage[margin=1in]{geometry}
\usepackage{tikz}
\usepackage{verbatim}
\usepackage{cancel}
\usepackage{titling}
\usepackage{multirow}
\usepackage{cancel}
\usepackage{quantikz}
\usepackage[colorinlistoftodos,prependcaption,textsize=tiny]{todonotes}
\usepackage{enumitem}
\usepackage{bbm} 
\usepackage{graphicx}

\newcommand{\kxor}{k\textsc{xor}}

\newcommand{\plant}{\mathsf{plant}}
\newcommand{\walk}{\mathsf{walk}}
\newcommand{\Unif}{\mathsf{Unif}}

\newcommand{\Null}{\textsc{Null}}
\newcommand{\Planted}{\textsc{Planted}}

\newcommand{\oddcolors}{\mathrm{oddcolors}}

\newcommand{\Ckr}{C_{\text{\tiny anti}}}

\newcommand{\symmdiff}{\bigtriangleup}

\newcommand{\rev}{{\mathsf{rev}}}

\renewcommand{\tilde}{\widetilde}

\title{A Classical Quadratic Speedup for Planted $k$\textsc{xor}}
\author{Meghal Gupta\thanks{UC Berkeley. \texttt{meghal@berkeley.edu}} \and William He\thanks{Carnegie Mellon University. \texttt{wrhe@cs.cmu.edu}} \and Ryan O'Donnell\thanks{Carnegie Mellon University. \texttt{odonnell@cs.cmu.edu}. Part of this work was done while the author was consulting for Google Quantum AI.} \and Noah G.~Singer\thanks{Carnegie Mellon University. \texttt{ngsinger@cs.cmu.edu}. Supported in part by an NSF Graduate Research Fellowship (Award DGE 2140739).}}
\begin{document}
\allowdisplaybreaks
\maketitle
\begin{abstract}
A recent work of Schmidhuber \emph{et al.}~(QIP, SODA, \& Phys.~Rev.~X~2025)
exhibited a quantum algorithm for the noisy planted $\kxor$ problem running quartically faster than all known classical algorithms.
In this work, we design a new classical algorithm that is quadratically faster than the best previous one, in the case of large constant~$k$.
Thus for such~$k$, the quantum speedup of Schmidhuber~\emph{et al.}\ becomes only quadratic (though it retains a space advantage).
Our algorithm, which also works in the semirandom case, combines tools from sublinear-time algorithms (essentially, the birthday paradox) and polynomial anticoncentration.
\end{abstract}

\section{Introduction}
Recall that a $k$-uniform hypergraph is a pair $\calH = (V,H)$, where $V$ is a set of vertices and
$H$ is a set of $k$-uniform hyperedges (meaning size-$k$ subsets of~$V$).
In this introduction, we fix the notation $n = |V|$ and $m = |H|$.

In this paper, we study the \emph{noisy planted $\kxor$ problem}
(also called \emph{sparse learning parities with noise (LPN)}), including in the ``semirandom'' case.

\begin{definition}\label{def:noisy planted kxor}
    Let $0 \le \rho \le 1$, $\calH = (V,H)$ be a $k$-uniform hypergraph, and $z \in \{\pm1\}^V$.
    The \emph{$\rho$-$\Planted$ $\kxor$ distribution} on $\calH$ and $z$,
    denoted $\plant^\rho_{\calH,z}$, is a distribution over hyperedge signs $\bb = (\bb_e)_{e \in H} \in \{\pm1\}^H$ (a.k.a.\ ``right-hand sides'') given as follows:
    For each $e \in H$ independently, sample $\bm{\eta}_e \in \{\pm1\}$ such that $\Ex[\bm{\eta}_e] = \rho$,
    and then output $\bb_e = \bm{\eta}_e \prod_{v \in e} z_v$.
\end{definition}
When $\rho = 0$, the bits $\bb_e$ are i.i.d.\ uniformly random in $\{\pm1\}$,
while when $\rho = 1$, each bit $\bb_e$ deterministically equals $\prod_{v \in e} z_v$.
We will be interested in the computational problem of distinguishing the cases $\rho = 0$ and $\rho$ a small positive constant,
which we refer to as the $\Null$ and $\Planted$ distributions, respectively.

The starting point for our work is a recent work of Schmidhuber \emph{et al.}~\cite{schmidhuber2025quartic}, which
exhibited a quantum algorithm for the noisy planted $\kxor$ problem running quartically faster than the previous best-known classical algorithms based on the ``Kikuchi method" and power method (see \Cref{sec:kikuchi intro} for more on this):
\begin{theorem}[\cite{schmidhuber2025quartic}]\label{thm:SOKB}
Let $n,k,m \in \N$ ($k \geq 4$ even), $0 < \rho \leq 1$.   Parameterize $m = \delta n^{k/2} \log n$. Let
    \[
        \ell = \max\Bigg\{~ \Bigg\lceil
    \underbrace{(2.01)^{\frac{2}{k-2}}}_{\text{prefactor}} \cdot 
    \underbrace{\left( \rho^2 \tbinom{k}{k/2}\right)^{-\frac{2}{k-2}}}_{\text{``$1/4$''}}  \cdot 
    \underbrace{(1/\delta)^{\frac{2}{k-2}}}_{\text{main}} \Bigg\rceil,~~ k~\Bigg\}.
    \]
    Assuming $\ell\leq O(\sqrt{n})$,
 there exists:
\begin{itemize}
    \item a classical algorithm (the power method) running in $\tilde{O}(n^\ell)$ time, and
    \item a quantum algorithm running in $n^{\ell/4 + k} \cdot \ell^{O(\ell)} \cdot \log^{\ell/2k + O(1)} n$ time,
\end{itemize}
both of which distinguish right-hand sides $\bb$ sampled from either the $\Null$ or $\rho$-$\Planted$ distributions w.p.\ $1 - n^{-\Omega(\ell)}$, given a uniformly random $n$-vertex, $m$-hyperedge, $k$-uniform hypergraph.
\end{theorem}
In this paper, we improve quadratically upon the best \emph{classical} algorithm, in the case of large constant~$k$.
This means that the quartic speedup over classical obtained by \cite{schmidhuber2025quartic} is merely quadratic in this regime.
Our main theorem, proven in \Cref{sec:distinguisher}, is the following:

\begin{theorem}[Corollary of \Cref{thm:main formal}]\label{thm:main informal}
    There exists a universal constant $C<\infty$ such that the following holds.
    Let $n,k,m\in \N$ ($k \geq 4$ even),  $0 < \rho \leq 1$. Parameterize $m = \delta n^{k/2} \log n$, and also write $\eps = \frac{10 \log(1/\rho)}{\log k}$. Let
    \[
        \ell = \max\Bigg\{~ \Bigg\lceil
    \underbrace{\pbra{\frac{C \log (1/\rho)}{\log k}}^{\frac{2}{k-2}}}_{\text{prefactor}} \cdot 
    \underbrace{\left( \rho^2 \tbinom{k}{k/2}\right)^{-\frac{2}{k-2}}}_{\text{``$1/4$''}}  \cdot 
    \underbrace{(1/\delta)^{\frac{2}{k-2}}}_{\text{main}} \Bigg\rceil , ~~k~\Bigg\}.
    \]
    Assuming $\ell\leq O(\sqrt{n})$, there is a classical algorithm that runs in $\tilde{O}\parens*{n^{(1/2 + \epsilon) \ell + k}}$ time\footnote{
    Here and elsewhere, $\tilde{O}(\cdot)$ hides polynomials in $\log m$ (note $\ell, n \le m$) and $\log (1/\rho)$.
    }
    and distinguishes right-hand sides $\bb$ sampled from either the $\Null$ or $\rho$-$\Planted$ distributions w.p.\ $1-n^{-\Omega(\eps \ell)}$,
    given a \emph{worst-case} $n$-vertex, $m$-hyperedge, $k$-uniform hypergraph.
\end{theorem}

The quantitative requirements on $\ell$ in \Cref{thm:SOKB,thm:main informal} are essentially formulations of the requirements needed for the $\ell$th level in the ``Kikuchi hierarchy" to yield a difference in the $\Null$ and $\Planted$ cases. For more explanation on these requirements see \Cref{sec:kikuchi intro}.

\paragraph{Comparisons.} Let us give some points of comparison between our \Cref{thm:main informal} and \Cref{thm:SOKB} from~\cite{schmidhuber2025quartic}.
\begin{itemize}
    \item Once $k \geq c_0 \log\log (1/\rho)$ (where $c_0$ is a universal constant), the ``prefactor'' in both theorems is between~$1$ and~$1.001$; i.e., the chosen $\ell$ value differs by only a factor of~$1.001$. As a consequence, the running time of the quantum algorithm is of the shape $\widetilde{O}(n^{.25 \ell + k})$, whereas our classical algorithm runs in time $\widetilde{O}(n^{.501 \ell + k})$.
    \item The competitiveness of our classical algorithm for small constant~$k$ depends essentially on the constant factor~$C$ achieved in our bounds. We did not try especially hard to optimize it, though at the same time we doubt our methods could reach, e.g., $C = 8.04$ for $k = 4$ (as would be needed to match \Cref{thm:SOKB}). As such, for small constant~$k$ (e.g., $k = 4$), the running time of the quantum algorithm still shows a quartic speedup over the best known classical algorithm.
    \item The work of~\cite{schmidhuber2025quartic} has the feature that any improved ``Kikuchi eigenvalue theorem'' (see \Cref{sec:kikuchi intro}) leads to a corresponding improved run-time.  As long as these improvements are confined to reducing the constant factor $2.01$ (to something at least $\exp(-o(k))$), the conclusion of our paper still holds: the quantum advantage is only quadratic, provided $k$ is a large enough constant.  On the other hand, eliminating the log factor --- i.e., showing a Kikuchi eigenvalue theorem merely with $\bar{d} \geq \frac{C}{\rho^2}$ (cf.~\Cref{eqn:req}) --- could restore the quartic advantage.  This is widely regarded as a major open problem, though, and it is plausible that if it were resolved, our method (based on finding short cycles in the Kikuchi graph) would also allow removal of the log factor.
    \item The work~\cite{schmidhuber2025quartic} has a space complexity advantage; it uses $\widetilde{O}(n^{k/2})$ classical space to prepare a quantum circuit, and that circuit uses $O(\ell \log n)$ qubits.  Our algorithm uses space proportional to its time.
    \item The quantum algorithm in  \cite{schmidhuber2025quartic} is only for the fully random case, meaning that the ``left-hand side'' hypergraph~$\calH$ is chosen randomly.  By contrast, our algorithm also works in the more challenging \emph{semirandom} model, where the hypergraph $\calH$ can be worst-case. This semirandom model, dating back to \cite{feige2007refuting}, was first handled by Kikuchi methods in~\cite{guruswami2022smoothed}, and has played an important role in recent developments in coding theory, in particular bounds on rates and distance of locally correctable and locally decodable codes~\cite{alrabiah2023near,kothari2024exponential}.
    It is possible that the quantum speedup from \cite{schmidhuber2025quartic} could be extended to work in the semirandom case, though we would not speculate on what the resulting prefactor constants might be.
    \item Finally, we emphasize that our faster classical algorithm is \emph{not} achieved by dequantizing the Guided Sparse Hamiltonian algorithm used in \cite{schmidhuber2025quartic}; rather, it comes from designing a classical algorithm related to cycles rather than eigenvalues. Thus obtaining quartic (or better) speedups for other problems via Guided Sparse Hamiltonian remains a viable possibility. 
\end{itemize}

\subsection{Kikuchi Graphs and Previous Work}\label{sec:kikuchi intro}
The problem of distinguishing the $\Null$ and $\Planted$ distributions has been studied in numerous works~\cite{wein2019kikuchi,hastings2020classical,schmidhuber2025quartic},
and hardness of similar problems forms the basis of security for cryptographic protocols (see e.g.~\cite{applebaum2010public,dao2023multi,dao2024lossy} among numerous other works).
The related problem of refuting the $\Null$ distribution was studied in e.g.~\cite{schoenebeck2008linear,barak2016noisy,allen2015refute,abascal2021strongly,guruswami2022smoothed,hsieh2023simple}.
Most typically the hypergraph $\calH$ and the ``planted assignment'' vector $z$ are also chosen uniformly at random.
In our work, as in \cite{abascal2021strongly,guruswami2022smoothed,hsieh2023simple}, we consider the more challenging \emph{semirandom} case in which only the right-hand sides are random.

Many algorithms for analyzing $\kxor$ problems $(\calH,b)$, including the one we will present,
involve constructing and analyzing the so-called \emph{level-$\ell$ Kikuchi graph} $K_\ell(\calH)$ for $\calH$. This graph/matrix, introduced by Wein \emph{et al.}~\cite{wein2019kikuchi} (and implicitly by Hastings~\cite{hastings2020classical}),
can be seen as a reduction from $\kxor$ to $\textsc{2xor}$,
allowing techniques suited to $\textsc{2xor}$ (such as spectral methods) to be applied.

\begin{definition}[Kikuchi graph]
Let $k$ be even,
$\calH = (V, H)$ a $k$-uniform hypergraph and $\ell \geq \frac{k}2$ an integer. (We usually imagine $\ell \gg k$.)
The \emph{level-$\ell$ Kikuchi graph} of $\calH$ is the graph $K_\ell(\calH)$ whose vertices are size-$\ell$ subsets of $V$
and which has an edge between $w_1, w_2$ iff $w_1 \symmdiff w_2 \in H$
(where $\symmdiff$ denotes the symmetric difference of sets).

More generally, given hyperedge-signs $(b_e)_{e \in H}$, one can also define the \emph{level-$\ell$ Kikuchi matrix} $M_\ell(\calH,b)$.
This is a signed version of the adjacency matrix of $K_\ell(\calH)$,
where the sign corresponding to an edge $\{w_1,w_2\}$ in $K_\ell(\calH)$ is $b_{w_1 \symmdiff w_2}$.
\end{definition}
\begin{remark}\label{rem:kikuchi params}
    The Kikuchi graph has $N_\ell \coloneqq \binom{n}{\ell} \leq (en/\ell)^\ell = O(n)^\ell$ vertices, and its average degree is
    $\bar{d}_\ell \coloneqq \frac{\binom{n-k}{\ell-k/2}}{\binom{n}\ell} \cdot \binom{k}{k/2} \cdot m$.
    The growth of this average degree in terms of~$\ell$ is quite important; provided $\ell = o(n/k)$, standard 
    estimates show 
    \[
        \bar{d}_\ell \sim \ell^{k/2} \cdot \Delta, \qquad \text{where } \Delta \coloneqq \tbinom{k}{k/2} \cdot \frac{m}{n^{k/2}}.
    \]
\end{remark}
With these definitions we can now go into some more detail about the quantitative requirements in \Cref{thm:SOKB,thm:main informal}.

The previous fastest (classical) algorithms for distinguishing  between $\Null$ and $\Planted$ distributions
were based on estimating the spectral norm of the Kikuchi matrix $M_\ell(\calH,b)$~\cite{wein2019kikuchi}.
The essential ingredient is what one might call a \emph{Kikuchi eigenvalue theorem}: namely, a theorem showing that once $\ell$ is sufficiently large as a function of $m, n, k, \rho$, there is (w.h.p.)\ a noticeable difference between the spectral norm of $M_\ell(\calH,b)$ in the $\Null$ and in the $\Planted$ cases.  Once such a theorem is proven for a certain value of~$\ell$, one can algorithmically estimate the spectral norm of $M_\ell(\calH,b)$ using the power method in $\wt{O}(n^\ell)$ (when the Kikuchi graph is sparse) time, thereby solving the distinguishing task. 
\Cref{thm:SOKB} from~\cite{schmidhuber2025quartic} is essentially a way to quartically speed up the spectral norm estimation using a quantum algorithm.

The best Kikuchi eigenvalue theorems previously proven (assuming $k$~even) essentially require
\begin{equation}
    \label{eqn:req}
    \bar{d}_\ell \geq \tfrac{C}{\rho^2} \cdot \log N_\ell
\end{equation}
for some universal constant~$C$.  A little painfully, this means that the exact exponent in the running time $\wt{O}(n^\ell)$ depends on the constant~$C$ established in the Kikuchi eigenvalue theorem.  On the other hand, asymptotically solving the above inequality for~$\ell$ yields:
\[
    \text{parameterizing $m = \delta n^{k/2} \log n$,} \qquad \eqref{eqn:req} \iff \ell \gtrsim 
    \underbrace{C^{\frac{2}{k-2}}}_{\text{prefactor}} \cdot 
    \underbrace{\left( \rho^2 \tbinom{k}{k/2}\right)^{-\frac{2}{k-2}}}_{\text{``$1/4$''}}  \cdot 
    \underbrace{(1/\delta)^{\frac{2}{k-2}}}_{\text{main}}.
\]
Here the ``main'' factor $(1/\delta)^{\frac{2}{k-2}}$ is the essential contributor to the value of~$\ell$.  The factor labeled ``$1/4$'' is a constant that depends on~$k$ (and~$\rho$), but it rapidly approaches the fixed value $1/4$ for large constant~$k$ (and constant~$\rho$).  Finally, in the  ``prefactor'', we see that obtaining the best possible constant factor~$C$ is arguably not so important after all, as this prefactor is exponentially suppressed towards~$1$ as $k$ becomes large.

While our algorithm does not rely on such a Kikuchi eigenvalue theorem, and thus cannot use the same requirement on $\ell$ as the one written above, we show that there exists a different $C$ such that \Cref{eqn:req} suffices for a classical algorithm running in time $\sim n^{\ell/2}$. As we just noted, this prefactor is suppressed with large $k$. Therefore, the requirement on $\ell$ is essentially the same as that required by \cite{schmidhuber2025quartic}, and our speedup is quadratic.

\subsection{Other Related Works}

\paragraph{Rainbow cycles.}
A cycle in an edge-colored graph is \emph{rainbow} if every color appears at most once.
Rainbow cycles are necessarily nontrivial (cf. \Cref{def:walk}).
In extremal combinatorics, the tradeoff between edge density and rainbow cycle length
has been studied in a number of works, e.g.~\cite{keevash2007rainbow,das2013rainbow,cada2016rainbow,chen2022rainbow,kim2024rainbow,janzer2024turan,alon2025essential}.
The hypergraph Moore bound was conjectured by Feige~\cite{feige2008small},
following work of Feige, Kim, and Ofek~\cite{feige2008refute} which connected even covers and refutations of random $\kxor$ instances.
The bound was then proved by Guruswami, Kothari, and Manohar~\cite{guruswami2022smoothed} and sharpened by Hsieh, Kothari, and Mohanty~\cite{hsieh2023simple};
both proofs still used spectral methods.
Most recently, Hsieh, Kothari, Mohanty, Correia, and Sudakov~\cite{hsieh2025even}
connected these two lines of work, giving a combinatorial proof of
the hypergraph Moore bound using rainbow cycle theorems.
They, however, must use linear-time algorithms in order to find their rainbow cycles.
Our \Cref{cor:moore bound} recovers their existential result (though with worse constants).

\paragraph{Related problems for $\kxor$.}
Numerous works have studied variants of the $\Null$-vs.-$\Planted$ distinguishing problem we study in this paper.
Firstly, algorithms for \emph{refuting} the $\Null$ distribution for $\kxor$
--- that is, certifying that randomly drawn instances lack assignments satisfying many constraints ---
were given in e.g. ~\cite{coja2007strong,allen2015refute,barak2016noisy,raghavendra2017strongly,abascal2021strongly,guruswami2022smoothed,hsieh2023simple} (see \cite{guruswami2022smoothed} for a full overview).
These algorithms are typically based on \emph{Sum-Of-Squares (SoS)} or other \emph{semidefinite programming (SDP)} methods. 
Lower bounds in these algorithmic models were given in e.g.~\cite{schoenebeck2008linear,kothari2017sum}.
These algorithms and lower bounds all rely on constraint density-runtime tradeoffs,
and the spectral algorithms based on Kikuchi matrices~\cite{wein2019kikuchi,hastings2020classical}
(essentially \Cref{thm:SOKB}) and the algorithm we present in this paper involve similar tradeoffs.
The problem of (approximately) outputting the planted assignment from the $\Planted$ distribution has also been studied,
e.g. in~\cite{feldman2015subsampled,guruswami2023efficient}.

Even covers appear in the algorithmic work~\cite{guruswami2022smoothed} as
efficiently \emph{verifiable} witnesses of unsatisfiability for random $\kxor$ instances.
Polynomial concentration also appears in the analysis of their algorithm,
although the inequalities they use assume more structure on the polynomials.

\paragraph{Sublinear algorithms.}
Sublinear-time testing algorithms for \textsc{Max-Cut}, which is a special case of \textsc{Max-2xor},
have been studied in numerous previous works. Some of these works, e.g.~\cite{goldreich1998sublinear,kaufman2004tight,peng2023sublinear,jha2024sublinear},
use a similar ``meet-in-the-middle'' approach for finding cycles that we take in our paper.
However, the \textsc{Max-Cut} setting is quite different from our own:
There, the goal is typically to find cycles of odd length,
whereas our goal is to find ``nontrivial'' cycles.
Further, the aforementioned works typically need strong assumptions on the input graph, such as bounded degree or even expansion.
In contrast, the Kikuchi graphs which we study are non-expanding.

\medskip

\section{Technical Overview}
Recall that $N_\ell = \binom{n}\ell$ is the number of vertices in the level-$\ell$ Kikuchi graph $K_\ell(\calH)$. Since our goal is to run in around $n^{\ell/2} \approx \sqrt{N_\ell}$ time, we must turn to the realm of \emph{sublinear algorithms}. 
Our algorithm cannot afford  even to write down the entire Kikuchi graph,
much less run the power method on it.
Consequently, our distinguisher does not (directly) estimate the spectral norm of the Kikuchi matrix $M_\ell(\calH,b)$.

Instead, the basis of our distinguisher is looking for a certain combinatorial structure, called an \emph{even cover}, in the $k$-uniform hypergraph $\calH$.
This high-level approach has been followed by several previous works (e.g.,~\cite{feige2008refute,feige2008small,guruswami2022smoothed})
for related $\kxor$ problems.
\begin{definition}[Even covers]
    Let $\calH = (V,H)$ be a $k$-uniform hypergraph.
    A subset $F \subseteq H$ of hyperedges is an
    \emph{even cover} of $\calH$ if $\bigtriangleup_{S \in F} S = \emptyset$
    or equivalently, for every $v \in V$, $\abs{\{S \in F : S \ni v\}} \equiv 0 \pmod 2$.
\end{definition}
In the setting of planted $\kxor$ (\Cref{def:noisy planted kxor}),
there is an important connection between even covers and parities of noise bits:
\begin{fact}\label{fact:even cover in planted distribution}
    If $F \subseteq H$ is an even cover of $\calH$,
    then $\prod_{e \in F} \bb_e = \pbra{ \prod_{e \in F} \bm{\eta}_e } \pbra{ \prod_{e \in F} \prod_{v \in e} z_v } = \prod_{e \in F} \bm{\eta}_e$.
\end{fact}
Indeed, $\prod_{e \in F} \bm{\eta}_e$ has expectation zero in the $\Null$ case
but expectation $\rho^{|F|}$ in the $\Planted$ case.
Hence $\prod_{e \in F} \bm{\eta}_e$ reveals information about whether we are in the $\Planted$ or $\Null$ case,
but the utility of this information degrades with $|F|$.
The tradeoff between hypergraph density and existence of short even covers is known as the \emph{hypergraph Moore bound}, recently established by \cite{guruswami2022smoothed,hsieh2023simple}.
At the densities we are interested in, the shortest even covers might have only logarithmic length,
meaning a single even cover cannot reveal enough information to distinguish between the cases.
Instead, we will end up needing to produce many such even covers.
This leaves us with two main technical challenges:
\begin{enumerate}
    \item \textbf{Finding covers:} How can we efficiently produce many short even covers in $\calH$?\label{item:find walks}
    \item \textbf{Using covers:} How can we combine these even covers to distinguish between the $\Null$ and $\Planted$ cases?\label{item:use walks}
\end{enumerate}
\Cref{item:find walks,item:use walks} are separate pieces of our algorithm and our solutions to each may be of independent interest.
\Cref{item:find walks} can be viewed as a sublinear-time constructive version of the hypergraph Moore bound proof~\cite{guruswami2022smoothed,hsieh2023simple,hsieh2025even},
is the only part of the proof that involves Kikuchi graphs,
and is the dominant contributor to our algorithm's runtime.

\subsubsection{Setup: Finding Small Even Covers by Finding Short Cycles}

Our first task is finding even covers in $\calH$ (\Cref{item:find walks}), which we do in \Cref{sec:finding walks} below.
We use a reduction to the problem of finding nontrivial closed walks in edge-colored graphs,
which has been used numerous times before, e.g. in \cite{hsieh2023simple,hsieh2025even}.
We first present the reduction and then turn to our algorithm for the reduced problem.

\begin{definition}
Let $G = (V,E)$ be a graph and $H$ a finite set.
An \emph{edge coloring} of $G$ with \emph{palette} $H$ is an assignment
$\chi : E \to H$ such that $\chi(e) \neq \chi(e')$ whenever $e \ne e' \in E$ are adjacent.
For a hypergraph $\calH = (V,H)$, the Kikuchi graph $K_\ell(\calH)$ is an edge-colored graph with palette $H$ and edge colors given by $\chi(\{w_1,w_2\}) = w_1 \symmdiff w_2$.
\end{definition}

\begin{definition}\label{def:walk}
    Let $G = (V,E)$ be a graph. A \emph{walk of length $T$} in $G$ is a sequence of vertices  $w_0,\ldots,w_T \in V$ such that for every $0 \le t < T$, $\{w_t,w_{t+1}\} \in E$.
    A walk $w_0,\ldots,w_T \in V$ is \emph{closed} if $w_T = w_0$.
    If $\chi : E \to H$ is an edge coloring and $W = (w_0,\ldots,w_T = w_0)$ be a closed walk, we let \[
    \oddcolors_\chi(W) \coloneqq \cbra{ C \in H ~:~\abs { \cbra { 0 \le t < T : \chi(\{w_t,w_{t+1}\}) = C } } \equiv 1 \pmod 2}
    \] denote the set of colors which occur an odd number of times along the walk.
    We say
    $W$ is \emph{trivial} if $\oddcolors_\chi(W) = \emptyset$ and \emph{nontrivial} otherwise.
\end{definition}
The key link between walks in the Kikuchi graph and even covers is:
\begin{fact}\label{fact:walks and even covers}
    Let $\calH = (V, H)$ be a $k$-uniform hypergraph and $\ell \in \N$. 
    Let $W = (w_0,w_1,\ldots,w_T=w_0)$ be a closed walk of length $T$ in the Kikuchi graph $K_\ell(\calH)$.
    Then $\oddcolors_\chi(w_0,\ldots,w_T)$ is an even cover of $\calH$
    (which is moroever nonempty if $w_0,\ldots,w_T$ is nontrivial).
\end{fact}
Hence finding many small even covers in a hypergraph reduces to
finding many short closed walks in edge-colored graphs
(where $\oddcolors$ is nonempty and different for each walk).

\subsubsection{Finding Short Nontrivial Walks}

In \Cref{sec:finding walks}, we focus on finding such short closed walks.
We first present an algorithm, \Cref{alg:find one walk}, for finding a single walk,
and then \Cref{alg:find many walks}, which repeatedly invokes \Cref{alg:find one walk} to find many walks.
We emphasize that these algorithms work generically for any edge-colored graph with sufficient average degree (not just Kikuchi graphs).

Given an edge-colored graph $(G = (V,E), \chi : E \to H)$ with $N = |V|$, \Cref{alg:find one walk} uses a ``meet-in-the-middle'' approach to find closed walks in roughly $\tilde{O}(\sqrt N)$ time.\footnote{We emphasize that $N = \binom{n}\ell$ in the context of the original Kikuchi graph.}
Informally, we have:

\begin{theorem}[Informal version of \Cref{thm:find walk}]
    Let $G = (V,E,\chi)$ be an $N$-vertex edge-colored graph of average degree~$d \ge \Omega(\ln N)$.
    Then \Cref{alg:find one walk} below (initialized at a random start vertex)
    runs in time $\tilde{O}(\sqrt N)$
    and, with probability $\Omega(1)$, returns a nontrivial closed walk of length $O(\ln N)$.
\end{theorem}

By \Cref{fact:walks and even covers}, we then have:

\begin{corollary}[``Constructive'' hypergraph Moore bound]\label{cor:moore bound}
    There exists a universal constant $1 < c < \infty$ such that the following holds.
    Let $\calH = (V,H)$ be a $k$-uniform hypergraph
    with $n$ vertices and $m$ hyperedges,
    with \[
    m \ge \parens*{ C \cdot \binom{k}{k/2}^{-1} } \cdot \parens*{ \frac{n}{\ell} }^{k/2} \cdot \ell \log n.
    \]
    Then $\calH$ contains a nonempty even cover of size $O(\ell \log n)$ (found efficiently by \Cref{alg:find one walk}).
\end{corollary}
Our approach to finding closed walks in \Cref{alg:find one walk}, is standard: if we run $\tilde{O}(\sqrt N)$ i.i.d.~random walks starting from a common vertex, by the birthday bound, with high probability two of them will have the same final vertex and therefore form a closed walk.
(This is where we save a $\sqrt{N}$ factor over more na\"ive algorithms.)

The main challenge here is to make sure that the closed walk that is found is \textit{nontrivial}. For this we note that intuitively, if two walks encounter many vertices of high degree, then it would be extremely unusual for them to use exactly the same colors an odd number of times. 
By a Markov-type argument, for a random start vertex $\bv$ with probability $\Omega(1)$, random walks from $\bv$ are \textit{good} in this sense (having many high degree vertices) with probability $\Omega(1)$.
Then, we use an encoding argument that it is unlikely that
two walks from a fixed start vertex $v$ are both good and
use the exact same colors an odd number of times.
Indeed, we show that stronger statement that the set of \emph{singleton} colors
--- the colors contained exactly once ---
is unlikely to be contained within any fixed small set.

This stronger property of avoiding fixed sets of singleton colors is
actually useful when we design \Cref{alg:find many walks},
our algorithm for outputting multiple walks.
\Cref{alg:find many walks} invokes \Cref{alg:find one walk} repeatedly,
and we use the avoidance property to guarantee that the corresponding even covers we get are all distinct.

\subsubsection{Distinguishing Using Even Covers}

After using the algorithms in the previous section to produce small even covers,
\Cref{alg:distinguisher} below uses these covers to distinguish the $\Null$ and $\Planted$ distributions.
Informally, we show the following:
\begin{lemma}[Informal version of \Cref{lem:distinguisher correct}]
    Let $\calH$ be a hypergraph, $\ell \in \N$, $N = \binom{n}\ell$, and $\rho > 0$.
    Then \Cref{alg:distinguisher}, given as input a set $\calC$ of nonempty even covers of size $\le O(\log N \poly(1/\rho))$ in $\calH$,
    with $|\calC| \ge 10N^\eps$,
    and a sample $\bb \sim \plant^\rho_{\calH,z}$
    outputs a guess \textsc{Null} or \textsc{Planted}, correct with probability at least 0.9.
\end{lemma}

Recall that for each even cover $F \in \calC$, the random variable $\prod_{e \in F} \bb_e$ has
expectation either $0$ in the $\Null$ case or $\rho^{|F|}$ in the $\Planted$ case.
Hence we can try to distinguish the two cases using the quantity $\bm{P} = \sum_{F \in \calC} \prod_{e \in F} \bb_e$.
Unfortunately, the summands are not independent since the $F$'s might share edges,\footnote{If we instead required the $F$'s to all be disjoint, we would run out of edges before getting a distinguisher.}
and so strong concentration bounds prove difficult.

In the $\Null$ case, a second moment calculation shows that $\bm{P}$ still does concentrate about zero,
because the signs on distinct even covers are pairwise independent
and by assumption we have collected a significant number of distinct even covers. 
Thus, all we need to show is that in the $\Planted$ case the polynomial $\bm{P}$ \textit{does not} concentrate about zero. 

To do this, we essentially only use two facts: that $\bm{P}$ is (a) low degree and (b) has mean far from zero.
However, when the inputs are Bernoulli random variables, there exist polynomials (indeed, linear ones) with mean bounded away from zero
which still have constant probability of evaluating to zero.
To get around this, we instead evaluate our polynomials on continuously noised versions of the variables $\bb_e$. In particular, we multiply $\bb_e$ by a uniform random value in $[0,1]$. A standard win-win argument now shows anticoncentration of any linear polynomial about zero:
\begin{itemize}
    \item If the variance of the polynomial output is small, then it concentrates about its mean, which is far from zero.
    \item If the variance of the polynomial output is large, then because the input random variables are not concentrated in any particular small interval,
    results such as those in \cite{miroshnikov1980inequalities,bobkov2014bounds} show that the polynomial's output also cannot be concentrated in any small interval, in particular a small interval about zero.
\end{itemize}
While our polynomial $\bm{P}$ is not necessarily linear, we can subsample the set of even covers (monomials) so that it is \textit{block-multilinear}, while not significantly decreasing the number of remaining even covers. We can then inductively apply the anticoncentration of linear polynomials to argue that the block-multilinear polynomial is anticoncentrated as well. Since there are still many remaining distinct monomials in the polynomial, the concentration argument in the $\Null$ case still holds, and this completes the proof of the correctness of the distinguisher. 

\section{Preliminaries}

We use $\Unif(S)$ to denote the uniform distribution over a set $S$
and $S \symmdiff T$ to denote the symmetric difference of two sets $S$ and $T$.
The logarithm $\log(\cdot)$ is base $2$ (while $\ln(\cdot)$ is base $e$ as usual).

\subsection{Graph Theory}

All graphs in this paper are undirected and simple.

Let $G = (V,E)$ be a graph.
The \emph{average degree} of $G$ is $\Ex_{v \sim \Unif(V)}[\deg(v)] = \frac{2|E|}{|V|}$.
The \emph{degree distribution} of $G$, denoted $\pi$,
is the distribution on $V$ given by $\pi(v) \propto \deg(v)$
(so that $\pi(v) = \frac{\deg(v)}{2|E|}$).
Note that this may not be a unique stationary distribution in the case where $G$ is disconnected,
but it is always stationary for the random walk on $G$.

\begin{fact}\label{fact:min degree}
    If $G$ has average degree $d$, then for any $0<\beta \leq 1$, it holds that $\Pr_{\bv\sim\pi}\sbra{\deg(\bv) < \beta d} < \beta$.
\end{fact}

\begin{proof}
    Let $S\coloneqq \{v:\deg(v)<\beta d\}$.
    Then
    \[
    \Pr_{\bv\sim\pi}[\deg(\bv)< \beta d]
    =\sum_{v\in S}\frac{\deg(v)}{2|E|}
    <\sum_{v\in S}\frac{\beta d}{2|E|}
    =\frac{\beta d|S|}{2|E|}
    \le\frac{\beta d|V|}{2|E|}
    =\beta.\qedhere
    \]
\end{proof}

\begin{definition}
    The \emph{random walk distribution of length $t$ from $w_0$}, denoted $\walk^T_{w_0}$,
    is a distribution over walks of length $t$ with initial vertex $w_0$ defined as follows:
    To sample $\bW = (\bw_0=w_0, \bw_1,\ldots,\bw_T) \sim \walk^T_{w_0}$,
    iteratively sample $\bw_{t+1}$ to be a random neighbor of $\bw_t$ for $1 \le t < T$.
\end{definition}

\begin{fact}\label{fact:stationary sample}
    Let $G = (V,E)$ be a graph and $T \in \N$.
    If $\bw_0 \sim \pi$ and $\bW = (\bw_0, \bw_1, \ldots, \bw_T) \sim \walk^T_{\bw_0}$,
    then the marginal distribution of $\bw_t$ is $\pi$ (for every $1 \le t \le T$).
\end{fact}

For a walk $W = (w_0,\ldots,w_T)$ in a graph, we let $W^{\mathsf{rev}} = (w_T,\ldots,w_0)$ denote the reversed walk.

\subsection{Probability}

We use the following well-known fact, which is essentially the famous Maclaurin inequality (see e.g.,~\cite[\S2.22]{hardy1934inequalities}):

\begin{fact}\label{fact:maclaurin}
    For every finite set $S$, distribution $\calD$ on $S$, and $k \in \N$,
    \[
    \Pr_{\bx_1,\ldots,\bx_k \sim \calD\text{ indep.}} \bracks*{ \exists\text{ collision among }\bx_1,\ldots,\bx_k} \ge \Pr_{\bx_1,\ldots,\bx_k \sim \Unif(S)\text{ indep.}}\bracks*{ \exists\text{ collision among }\bx_1,\ldots,\bx_k}.
    \]
\end{fact}

\section{Finding Closed Walks in an Edge-Colored Graph}\label{sec:finding walks}

We now describe an efficient randomized procedure
for producing many nontrivial closed walks $W$ in an edge-colored graph,
where all the sets $\oddcolors_\chi(W)$ are distinct.

\subsection{Finding a Single Closed Walk}

The following algorithm refers to the notion of a walk in a graph being ``good''.  
We will later pick a specific ``goodness'' property, but at first we analyze things generically.

\begin{algorithm}[H]
\caption{\textsc{Find-Good-Closed-Walk}$(G,v_0, T) \\  \hspace*{2.5in} // \ G = (V,E)$ an $N$-vertex graph, $v_0 \in V$, $T \in \N$}\label{alg:find one walk}
    Draw $\bW_1, \dots, \bW_L \sim \walk^T_{v_0}$ independently, $L \coloneqq \lceil C_1 \sqrt{N} \rceil$.

    \If{$\exists\ i \neq j \in [L]$ \textnormal{such that $\bW_i, \bW_j$ are both ``good'' and end at the same vertex}}{
    \Return{$(\bW_i, \bW_j^\rev)$}}
    
    \Return{``fail''}
\end{algorithm}

\medskip
In the following, we will eventually take $\eps = .04$ and $C_1 = 200$ (hence we need $N > 40$):
\begin{lemma}   \label{lem:A}
    Assume that for $\bW \sim \walk^T_{v_0}$ we have $\Pr[\bW\ \text{good}] \geq \eps$. 
    Then \textsc{Find-Closed-Walk}$(G,v_0,T)$ successfully returns a closed walk (with both halves good) except with probability at most $2\exp(-\frac{\eps^2 C_1^2}{10})$. 
    (This assumes $L \geq .8 \eps C_1^2,\;30/\eps$.)
\end{lemma}
\begin{proof}
    By a Chernoff bound, the number of good walks drawn in Step~1 is at least $L' \coloneqq \lfloor (\eps/2) L\rfloor$ except with probability at most 
    \[
      \exp(-\eps L/8) \leq 2\exp\parens*{-\frac{\eps^2 C_1^2}{10}}
    \]
    (where we used $L \geq .8 \eps C_1^2$).
    Losing this probability, we may as well assume we have $L'$ independent draws of $\bW \sim \walk^T_{v_0}$ conditioned on $\bW$ being good.  
    Let $\calD$ denote the distribution on final vertices achieved by this conditioned distribution.
    If $\calD$ were the uniform distribution on~$V$, we would precisely have an instance of the Birthday Problem.  
    But in fact, it is well known (cf. \Cref{fact:maclaurin}) that the probability of collision in $L'$ draws from a distribution~$\calD$ is minimized when $\calD$ \emph{is} uniform. 
    Thus we may use the standard (uniform) Birthday Problem analysis (e.g., \cite{sayra1966birthday}) to conclude that the probability of getting no collisions in $L'$ draws from a distribution on $N$ elements is at most
    \[
        \exp\parens*{-\frac{L'(L'-1)}{2N}} 
        \leq \exp\parens*{-\frac{\eps^2 L^2}{10N}} 
        \leq \exp\parens*{-\frac{\eps^2 C_1^2}{10}}
    \]
    (where the first inequality  used $\eps L \geq 30$), completing the proof.
\end{proof}

Next we will develop the goodness property of walks we need, which will be that they hit few vertices with abnormally low degree.
\begin{definition}
    Let $0 < \beta \leq 1$.  In a graph $G = (V,E)$ of average degree $d$, we say that $u \in V$ is \emph{$\beta$-bad} if $\deg(u) < \beta d$.
    For a walk $W = (w_0, \dots, w_T)$, we write $B_\beta(W) = \#\{0 \leq i < T : w_i \text{ is $\beta$-bad}\}$. Finally, we say $W$ is \emph{$\beta$-good} if $B_\beta(W) \leq 1.1 \beta T$.
\end{definition}
\begin{lemma} \label{lem:B}
    In the previous setting, say that vertex $v_0 \in V$ is an \emph{$\eps$-good start} if the following condition holds: $\Pr_{\bW \sim \walk^T_{v_0}}[\bW \text{ is $\beta$-good}] \geq \eps$.  Then 
    $
        \Pr_{\bv \sim \pi}[\bv \text{ is a $.04$-good start}] > .04.
    $
\end{lemma}
\begin{proof}
    Suppose $\bv \sim \pi$ and then $\bW = (\bw_0, \dots, \bw_T) \sim \walk^T_{\bv}$. By \Cref{fact:stationary sample}, each $\bw_i$ is marginally distributed according to~$\pi$.  Thus by \Cref{fact:min degree} we have $\Pr[\bw_i \text{ is bad}] < \beta$ for each~$i$.  It follows that $\E_{\bv} \E_{\bW}[|B_\beta(\bW)|] < \beta T$. Thus $\Pr_{\bv \sim \pi} \bracks*{ \E_{\bW \sim \walk^T_{\bv}}[B_\beta(\bW)] < 1.05 \beta T } > 1 - \frac1{1.05} > .04$ by Markov's inequality.
    Finally, suppose $v_0$ indeed has $\E_{\bW \sim \walk^T_{v_0}}[B_\beta(\bW)] < 1.05 \beta T$.  Then Markov's inequality implies that $\Pr_{\bW \sim \walk^T_{v_0}}[B_\beta(\bW) > 1.1 \beta T] < \frac{1.05}{1.1} < .96$: i.e., $v_0$ is a $.04$-good start.
\end{proof}

Suppose now that we have an edge-colored graph of average degree $d \gg T$, and that $v_0$ is a good start vertex.  Then intuitively, length-$T$ walks from $v_0$ are likely to encounter many vertices of high degree, and thus are likely to exhibit a high diversity of colors.
The below key lemma makes this precise:
\begin{definition}
    Let $G = (V,E)$ be an edge-colored graph with coloring $\chi : E \to H$ and average degree~$d$. Let $W^1, W^2$ be two walks in $G$, and write $\calE = \calE(W^1, W^2) = (e_1, \dots, e_T, e_{T+1}, \dots, e_{2T})$ for the sequence of edges in walk $W^1$ followed by those in walk $W^2$. We then define $S(W^1, W^2)$ to be the \emph{singleton} edge-colors in~$\calE$; i.e., the set of colors that appear exactly once in the list $(\chi(e_i))_{i = 1\dots 2T}$.
\end{definition}
\begin{lemma} \label{lem:C}
    In the previous setting, fix any $v_0 \in V$ and any ``avoidance'' set of colors~$A \subseteq H$ satisfying $|A| \leq 2T$.\footnote{There is nothing special about $2T$ here; we could have any $O(T)$ by changing the constant $400$ in the lemma's conclusion.}
    Suppose we independently draw $\bW^1, \bW^2 \sim \walk^T_{v_0}$, and define the event
    \[
        \bcalC = \text{``}\bW^1, \bW^2 \text{ are $\beta$-good, and } S(\bW^1, \bW^2) \subseteq A\text{''.}
    \]
    Then fixing, say, $\beta = .05$, we have $\Pr[\bcalC] \leq (400 T/d)^{.89T}$.
\end{lemma}
\begin{proof}
    For a given $\calE = \calE(W^1,W^2) = (e_1, \dots, e_{2T})$, write $e_i = (w_{i-1}, w'_i)$, where the edge is naturally directed from its walk.\footnote{Note that $w_i = w'_i$ for all $i$ except possibly $i = T$, wherein $w_T$ equals $w_0$ but not necessarily $w'_T$. These observations do not affect our argument; all we'll need is that conditioned on any $\bw_{i-1} = v$, the distribution of $\chi(\be_i)$ is uniform on a set of size $\deg(v)$.}
    
    For $1 \leq i \leq 2T$, let us say that the $i$th time step is \emph{unusual} if:
    \begin{enumerate}[label=\roman*.]
        \item $w_{i-1}$ is \emph{not} $\beta$-bad; yet,
        \item $\chi(e_i) \in \{\chi(e_1), \dots, \chi(e_{i-1})\} \cup A$ (i.e., $\chi(e_i)$ is either a ``repeat'' or is in the avoidance set).
    \end{enumerate}
    Now define the \emph{unusualness string} $U(\calE) \in \{0,1\}^{2T}$ to have $U(\calE)_i = 1$ iff the $i$th step of $\calE$ is unusual.

    We make the following claim:
    \begin{equation}
        \label{eqn:theclaim}
        W^1, W^2 \text{ are $\beta$-good, and } S(W^1, W^2) \subseteq A \quad\implies\quad |U(\calE)| \geq (1-2.2\beta)T = .89T
    \end{equation}
    (where $|{\cdot}|$ denotes Hamming weight).  To justify the claim, first suppose that $W^1, W^2$ have $S(W^1, W^2) \subseteq A$.
    Then at most half of all time steps $i \in [2T]$ can fail condition~(ii).  (For each step color that is not in~$A$ and not a repeat, a future step must repeat this color.) Thus at least $T$ time steps satisfy~(ii).  Suppose now that $W^1, W^2$ are also $\beta$-good.  Then by definition, at most $2.2 \beta T$ time steps fail condition~(i).

    Having justified \Cref{eqn:theclaim}, we may conclude
    \begin{equation} \label{eqn:put}
        \Pr[\bcalC] = \sum_{u \in \{0,1\}^{2T}} \Pr[\bcalC \text{ and } U(\bcalE) = u] =  \sum_{|u| \geq
        .89T}\Pr[\bcalC \text{ and } U(\bcalE) = u] \leq \sum_{|u| \geq 
        .89T}\Pr[U(\bcalE) = u].
    \end{equation}
    Let us now fix some $u \in \{0,1\}^{2T}$ with $|u| \geq 
    .89T$, and bound $\Pr[U(\bcalE) = u]$.  
    Pick any time step $j$ for which $u_j = 1$.  We claim that
    \begin{equation}
        \label{eqn:confirm}
        \Pr[U(\bcalE)_{j} = u_{j} = 1 \mid U(\bcalE)_{< j} = u_{<j}] \leq \frac{j + |A|}{\beta d} \leq \frac{2T + 2T}{.05 d} = 80T/d.
    \end{equation}
    To see this, imagine the steps of the walks $\bW^1, \bW^2$ being drawn in succession and suppose we condition on $U(\bcalE)_{< j} = u_{<j}$.  Certainly this could greatly distort the distribution of the first $j-1$ steps of the random walk(s).  However, even under this conditioning, for any outcome $\bw_{j-1} = v$, the distribution of $\bw'_i$ is still uniformly random from among $v$'s neighbors. That is, $\chi(\be_i)$ is uniformly random on a set of size~$\deg(v)$.
    Now in order for $U(\bcalE)_j = 1$ to occur, we must have condition~(i) of unusualness: that $\bw_{j-1} = v$ is not $\beta$-bad; i.e., that $\deg(v) \geq \beta d$.  But we also must have condition~(ii), that $\chi(\be_j)$ falls into $\{\chi(\be_1), \dots, \chi(\be_{j-1})\} \cup A$, a set of cardinality at most $j + |A|$. But if $\chi(\be_j)$ is uniform on a set of size at least $\beta d$, the chance of this is indeed at most $\frac{j+|A|}{\beta d}$, confirming \Cref{eqn:confirm}.

    Now using $\Pr[U(\bcalE) = u] = \prod_{i} \Pr[U(\bcalE)_i = u_i \mid U(\bcalE)_{< i} = u_{< i}]$, we conclude from \Cref{eqn:confirm} that
    \begin{equation}
        \Pr[U(\bcalE) = u] \leq \parens*{80T/d}^{|u|}.
    \end{equation}
    (We may assume the base of the exponent is at most~$1$, else there's nothing to prove.) 
    Putting this into \Cref{eqn:put} completes the proof:
    \[
        \Pr[\bcalC] \leq 2^{2T} (80d/T)^{.89T} \leq (400d/T)^{.89T}. \qedhere
    \]
\end{proof}

Putting together \Cref{lem:A,lem:B,lem:C}, we conclude the following:
\begin{theorem}\label{thm:find walk}
    There is a universal constant $c < \infty$ such that the following holds.
    Let $G = (V,E,\chi)$ be an $N$-vertex edge-colored graph of average degree~$d$, where we assume 
    \[
        T \ln\parens*{\frac{d}{400T}} \geq \frac{1}{.89} \ln N + c
    \]
    (and $N > 40$).
    Let $A$ be a set of colors with $|A| \leq 2T$.
    Suppose we choose $\bv \sim \pi$ and then run \textsc{Find-Good-Closed-Walk}$(G,\bv, T)$ (with its $C_1 = 200$).  
    Then with probability at least~$.02$, it returns a closed walk $\bcalE$ with $S(\bcalE) \not \subseteq A$.
\end{theorem}
\begin{proof}
    With $\beta = .05$, \Cref{lem:B} tells us that $\bv$ is a $.04$-good start with probability exceeding~$.04$.  Fix any such outcome $\bv = v_0$.
    Then, with $\eps = .04$ and $C_1 = 200$, \Cref{lem:A} tells us that \textsc{Find-Good-Closed-Walk}$(G,v_0, T)$ will return some closed walk $\bcalE = (\bW_i, \bW_j^\rev)$ with both $\bW_i, \bW_j$ being $\beta$-good, except with probability at most $2\exp(-.04^2 \cdot 200^2 / 10) \leq .01$. Finally, by union-bounding over all the at most $\binom{L}{2} \leq O(N)$ pairs of $\beta$-good $\bW_i, \bW_j$, \Cref{lem:C} tells us that $S(\bW_i, \bW_j) \subseteq A$ with probability at most
    $
        O(N) \cdot (400T/d)^{.89T} $, which is at most $.01$ provided the constant~$c$ is large enough.
     As $.04 - .01 -.01 = .02$, the proof is complete.
\end{proof}

\subsection{Finding Many Distinct Nontrivial Closed Walks}

Building on the previous section, we now consider an algorithm to produce many closed walks:

\begin{algorithm}[H]
\caption{\textsc{Find-Distinct-Nontrivial-Closed-Walks}$(G,\chi,\delta,T,R)$}\label{alg:find many walks}
    \KwData{An undirected graph $G=(V,E)$, an edge-coloring $\chi : E \to H$, failure parameter $\delta$, average degree $d$, and number of desired cycles $R$.}
    \For{$r=1,\ldots,R\coloneqq100000N^{\eps} + 100000\log(1/\delta)$}{
        Run \textsc{Find-Closed-Walk}$(G,\bv,T)$ for $\bv \sim \pi$\;
        If it produces a nontrivial closed walk $W^r$, and $\oddcolors_\chi(W^r) \ne \oddcolors_\chi(W^{r'})$ for all $r' < r$, add $W^r$ to the output stream.
    }
\end{algorithm}

\begin{lemma}\label{lem:find many walks}
    Let $(G = (V,E), \chi : E \to H)$ be an edge-colored graph.
    Fix a positive integer $T\geq 100$ and $\eps, \delta \in (0,0.1)$. Suppose that $G$ has average degree  $d\geq 300 N^{4/T}T$.
    Then with probability at least $1-\delta$, \textsc{Find-Distinct-Nontrivial-Closed-Walks}$(G,\chi,\delta,T,\eps)$
    outputs a list of $\lceil 10N^{\eps} \rceil$ nontrivial closed walks of length $\le 2T$,
    with $\oddcolors_\chi(W)$ distinct for each walk $W$.
\end{lemma}
\begin{proof}
    For all $r\in[R]$, let $\calE_k$ be the event that iteration $R$ of \Cref{alg:find many walks} produces a closed walk that is nontrivial and no previous closed walk uses the same set of edges an odd number of times. Assuming we have not yet found $\lceil 10N^{\eps} \rceil$ distinct nontrivial closed walks, the probability that this event occurs is at least the probability that \textsc{Find-Closed-Walk}$(G,\chi,T)$ returns a closed walk minus the probability that the walk is trivial or that exact set of edges occurred an odd number of times in one of the previous $\leq \lceil 10N^{\eps} \rceil \leq N^{0.11}$ walks. The first event occurs with probability at least $0.01$, while the second event occurs with probability at most $(1+N^{0.11})\cdot 0.01/N$ by \Cref{thm:find walk}, and so 
    \[
        \Pr \sbra{ \text{$r$-th iteration outputs a new closed walk}} \geq 0.01 - (1+N^{0.11})\cdot 0.01/N \geq 0.005.
    \]
    Now, let $X_i\in\{0,1\}$ be the indicator that the $r$'th step produces a new closed walk, or there are already $\lceil 10N^{\eps} \rceil$ closed walks. Then,
    $\E[X_i\mid X_{1:i-1}]\ge 0.005$.  Hence
    \[
        \mu \;=\;\E\Bigl[\sum_{i=1}^R X_i\Bigr]\;\ge\;0.005R
        \;=\;500\,N^\eps + 500\log(1/\delta).
    \]
    Since each $X_i-p$ is a submartingale with difference bounded in $[-p,1-p]\subset[-1,1]$, Azuma–Hoeffding gives for $t=490\,N^\eps + 490\log(1/\delta)$,
    \begin{multline*}
        \Pr\sbra{\sum_{i=1}^R X_i < \lceil 10N^\eps \rceil}
        \leq\Pr\sbra{\sum X_i \le \mu - t} \\
        \le \exp\!\Bigl(-\frac{t^2}{2R}\Bigr)
        = \exp\!\Bigl(-\frac{490^2(N^\eps+\log(1/\delta))^2}{200000 (N^\eps+\log(1/\delta)}\Bigr) \leq \exp(-N^\eps - \log(1/\delta)) \leq \delta.
    \end{multline*}
    Therefore, with probability at least $1-\delta$, the algorithm outputs $\lceil 10N^\eps \rceil$ distinct closed walks.
\end{proof}

\section{Polynomial Anticoncentration Theorem} \label{sec:polynomial}

    \begin{definition}[$\pi$-block-linearity]
        Fix integers $m,T$, let 
        \[
            P(x_1,\dots,x_m) = \sum_{S\subseteq[m]}a_S\prod_{i\in S}x_i
        \]
        be a multilinear polynomial, and let $\pi=(\pi_1,\dots,\pi_T)$ be a partition of $[m]$.
        For $S \subseteq [m]$, we say $\pi$ \emph{shatters} $S$ if $|S \cap \pi_j| \le 1$ for every $j \in [T]$.
        Define the block-linear restriction of $P$ to $\pi$ as
        \[
            P_{\pi}(x_1,\dots,x_m)
            \coloneqq\sum_{\substack{S\subseteq[m],\pi\text{ shatters } S}}a_S\prod_{i\in S}x_i.
        \]
        We say that $P$ is \emph{$\pi$-block-linear} if $P_{\pi}=P$.
    \end{definition}

    We first show that a \emph{random} block-linear restriction $P_{\bm{\pi}}$ of a polynomial $P$ with all nonzero coefficients equal to 1 has at least some small probability of preserving a small fraction of terms from $P$.

    \begin{lemma}[Block‐restricted polynomial]\label{lem:block}
        Let $m,T\in\mathbb{N}$ with $T\mid m$, and let $P(x_1,\ldots,x_m)$ be a sum of $\mu$ multilinear monomials in the $x_i$'s of  degree at most~$T$.
        Then for a uniformly random equipartition $\bm{\pi} = (\bm{\pi}_1,\ldots,\bm{\pi}_T)$ of~$[m]$, 
        \[
            \Pr\sbra{\text{number of monomials in } P_{\bm{\pi}} \geq 0.5\cdot \mu \cdot \exp(-T)} \ge 1-\frac{1-\exp(-T)}{1-0.5\exp(-T)} \geq \exp(-2T).
        \]
    \end{lemma}
    \begin{proof}
        We will show, for each fixed monomial $x^S \coloneqq \prod_{i \in S} x_i$ in~$P$, that its probability of appearing in $P_{\bm{\pi}}$ is at least $\exp(-T)$.
        Thus the expected fraction of $P$'s monomials appearing in $P_{\bm{\pi}}$ is at least $\exp(-T)$, and then the lemma follows from a Markov-like argument on the number of monomials that \textit{do not appear}.

        Monomial $x^S$ appears in $P_{\bm{\pi}}$ iff $\bm{\pi}$ shatters $S$. The probability of this only depends on $|S|$; it is clearly minimized when $|S| = T$, so we henceforth assume this.  Also, we may equivalently fix an equipartition $\pi$, and then compute the probability that a uniformly random subset $\bS \subseteq [m]$ of cardinality~$T$ is shattered by $\pi$.

        Suppose for the moment that the $T$ elements of $\bS$ were drawn randomly with replacement, rather than without.  Then we would  have
        \[
            \Pr[|\bS \cap \pi_i| = 1 \ \forall i] = \frac{T}{T} \cdot \frac{T-1}{T} \cdot \frac{T-2}{T} \cdots \frac{2}{T} \cdot \frac{1}{T} = \frac{T!}{T^T} \geq 2\exp(-T).
        \]
       But it's easy to see that $\bS$ being drawn \emph{without} replacement only increases this probability.
    \end{proof}

\begin{lemma}\label{lem:linear anticoncentration}
    Let $\bx_1, \dots, \bx_n$ be independent ``nice'' random variables, where ``nice'' means that each $\bx_i$ is mean-zero, supported on $[-2,2]$, and with pdf bounded by~$1$.
    Let $\bL = c_1 \bx_1 + \cdots + c_n \bx_n$, where $c_1, \dots, c_n \in \R$.
    Then for all $0 < \eps < 1/2$, we have $\Pr[|\bL - 1| \leq \eps] \leq \Ckr \eps\sqrt{\log(1/\eps)}$, where $1.2<\Ckr < \infty$ is a universal constant.
\end{lemma}
\begin{remark}
    It is likely the bound can be improved to~$O(\eps)$.
    Also, if the $\bx_i$'s are particularly nice --- say, the sum of a bounded random variable and a standard Gaussian, which would suffice for us --- then the below proof becomes elementary.
\end{remark}
\begin{proof}
    Let $\sigma = \sqrt{\sum_i c_i^2}$.
    Since the $\bx_i$'s are bounded in $[-2,2]$, they are subgaussian with variance proxy $O(1)$ (where here and throughout the $O(\cdot)$ hides only universal constants). Therefore $\bL$ is subgaussian with variance proxy $O(\sigma^2)$.  It follows that for some $c > 0$, if $\sigma < c/\sqrt{\log(1/\eps)}$, we already have $\Pr[\bL \geq 1/2] \leq \eps$.  Thus we may henceforth assume $\sigma \geq c/\sqrt{\log(1/\eps)}$.  Now we replace $\bL$ by $\bL/\sigma$, so that $\sum_i c_i^2 = 1$ and we need to upper-bound the probability that $\bL$ is in a certain interval of width $2\eps/\sigma \leq O(\eps \sqrt{\log(1/\eps)})$.  It thus suffices to show $\bL$ has pdf bounded by~$O(1)$. But this follows directly from~\cite[Corollary~1]{miroshnikov1980inequalities}; or, even more directly, from~\cite[Corollary~2]{bobkov2014bounds}. (It is also immediate if the $\bx_i$'s are mixtures of Gaussians, as mentioned.)
\end{proof}
    
\begin{corollary}\label{cor:affine anticoncentration}
    Let $\bx_1, \dots, \bx_n$ be independent random variables that are mean $\rho$, supported on $[-1,1]$, and with pdf bounded by~$1$.
    Let $L$ be an affine polynomial on $n$ variables.
    Then for all $0 < \eps < 1/2$, we have $\Pr[|L(\bx_1,\ldots,\bx_n)| \leq \eps L(\rho,\ldots,\rho)] \leq \Ckr\eps\sqrt{\log(1/\eps)}$ (where $\Ckr$ is the constant from \Cref{lem:linear anticoncentration}).
\end{corollary}

\begin{proof}
    By affineness, $L(x_1,\ldots,x_n) = L(x_1-\rho,\ldots,x_n-\rho) + L(\rho,\ldots,\rho) - L(0,\ldots,0)$.
    Noting that $L(x_1-\rho,\ldots,x_n-\rho) - L(0,\ldots,0)$ is a homogeneous (linear) polynomial,
    by \Cref{lem:linear anticoncentration},
    \[
    \Pr \bracks*{ \abs { \frac{L(\bx_1-\rho,\ldots,\bx_n-\rho)-L(0,\ldots,0)}{-L(\rho,\ldots,\rho)}-1 } \le \epsilon } \le \Ckr\epsilon\sqrt{\log(1/\epsilon)}.
    \]
    This event is equivalent to the desired event.
\end{proof}

    Next, we extend the previous result to $\pi$-block-linear polynomials.

    \begin{proposition}[Anti‐concentration for block-linear polynomials]\label{prop:block linear anticoncentration}
        Let $\calD$ be a distribution on $\mathbb{R}$ with probability density $f_\calD\le1$, mean $\rho>0$ and support contained in $[-1,1]$. Fix integers $m,T\ge1$, a partition $\pi = (\pi_1,\ldots,\pi_T)$ of $[m]$, and let $P(x_1,\dots,x_m)$
        be any $\pi$-block-linear polynomial of degree $\le T$. Define
        \[
            \mu \coloneqq \underset{\bx_i\overset{\mathrm{iid}}{\sim}\calD}{\E} \sbra{P(\bx_1,\dots,\bx_m)} = P(\rho,\rho, \ldots, \rho).
        \]
        Then:
        \[
            \Pr\sbra{ \big|P(\bx_1,\dots,\bx_m)\big|<\Ckr^{-10T}|\mu| } \le 0.9^T.
        \]
    \end{proposition}
    \begin{proof} 
        Without loss of generality assume that $\pi_t=\left\{(t-1)\frac{m}{T}+1,\dots, t\cdot\frac{m}{T}\right\}$. Given fixed $x_1,\ldots,x_m \in \R$ and $t \in [T]$, define
        \begin{align*}
            r_t &= \Ex_{\by_{t\frac{m}{T}+1},\dots,\by_m\sim\calD}\sbra{P\pbra{x_1,\dots,x_{(t-1)\frac{m}{T}},\by_{(t-1)\frac{m}{T}+1},\dots,\by_m}}
            =P\pbra{x_1,\dots,x_{(t-1)\frac{m}{T}},\rho,\dots,\rho}.
        \end{align*}
        Note that $r_0$ is deterministically equal to $\mu$ (no dependence on $x$'s) and $r_T$ is $P(x_1,\ldots,x_m)$.

        Now, we consider what happens when we draw $\bx_1,\ldots,\bx_m \sim \calD$ i.i.d.,
        and therefore $\br_1,\ldots,\br_T$ are also random variables (depending on the $\bx$'s).
        We are trying to bound the probability of the following event, which can be described in several equivalent ways:
        \begin{multline*}
         \abs{P(\bx_1,\ldots,\bx_m)} < \Ckr^{-10T} |\mu| \iff |\br_T| < \Ckr^{-10T} |\br_0|  \\
        \iff -\log (|\br_T|/|\br_0|) > 10T \log\Ckr \iff \sum_{t=0}^{T-1} \bs_t > 10T \log\Ckr \iff 2^{0.1 \sum_{t=0}^{T-1} \bs_t} > \Ckr^T,
        \end{multline*}
        where we define $\bs_t \coloneqq -\log(|\br_{t+1}|/|\br_t|)$ and telescope $\sum_{t=0}^{T-1} \bs_t = -\log |\br_T|/|\br_0|$.
        Also note that 
        $\bs_t \ge \theta \iff |\br_{t+1}| \le 2^{-\theta} |\br_t|$.

        Now condition on $\bx_1,\ldots,\bx_{(t-1)\frac{T}m} = x_1,\ldots,x_{(t-1)\frac{T}m}$.
        This completely determines $\br_0,\ldots,\br_t = r_0,\ldots,r_t$ and therefore $\bs_0,\ldots,\bs_{t-1} = s_0,\ldots,s_{t-1}$.
        Further, $\br_{t+1}$ is now an affine polynomial in the variable
        $\bx_{t\frac{T}m+1},\ldots,\bx_{(t+1)\frac{T}m}$
        with mean $\br_t$.
        Hence by \Cref{cor:affine anticoncentration},
        $\Pr[s_t \ge \theta \mid x_1,\ldots,x_{(t-1)\frac{T}m}] \le \Ckr 2^{-\theta} \sqrt{\theta}$.
        For each $t\le T$, conditioning on the past and using the exponential tail bound of \Cref{cor:affine anticoncentration} gives
        \[
        \begin{aligned}
            \Ex\left[2^{0.1\bs_t}\mid x_1,\ldots,x_{(t-1)\frac{T}m} \right]
            &= \int_{0}^{\infty}
                 \Pr\sbra{ 2^{0.1\bs_t}>y\mid x_1,\ldots,x_{(t-1)\frac{T}m} }dy \\
              &\le 1 + \int_{1}^{\infty}
                 \Pr\sbra{\bs_t>10\log y \mid x_1,\ldots,x_{(t-1)\frac{T}m}  }dy\\
            &\le 1 + \Ckr \int_{1}^{\infty}
               2^{-10 \log y} \sqrt{10 \log y} \; dy
              \le 1 + 0.04 \Ckr.
        \end{aligned}
        \]
        Hence
        \[
            \Ex\left[2^{0.1(\bs_0+\cdots+\bs_{T-1})}\right]
            \leq \prod_{t=1}^T
            \left(
            \sup_{x_1,\ldots,x_{(t-1)\frac{T}m} }
            \Ex\left[2^{0.1\bs_t}\mid x_1,\ldots,x_{(t-1)\frac{T}m} \right]
            \right)
            \le (1+0.04\Ckr)^T.
        \]
        Finally, Markov’s inequality yields
        \[
        \begin{aligned}
            \Pr\sbra{2^{0.1(\bs_0+\cdots+\bs_{T-1})} > \Ckr^T} \le \frac{\mathbb{E}\left[2^{0.1(\bs_0+\cdots+\bs_{T-1})}\right]}
                       {(\Ckr)^T}
            \le\left(\frac{1 + 0.04\Ckr}{\Ckr}\right)^T
            <0.9^T
        \end{aligned}
        \]
        using $\Ckr > 1.2$, as desired.
    \end{proof}

\section{Distinguishing Using Even Covers}\label{sec:distinguisher}

In this section, we show how to take a sufficiently large set of even covers of a hypergraph $\calH$
and use this to distinguish between the $\Null$ and $\Planted$ cases.
Suppose we have a set $\calC$ of even covers of $\calH$.
We have at least $|\calC| \ge 10N^\eps$ even covers,
and each cover $C \in \calC$ has length $|C| \le 2T$.
Now associate to $\calC$ the polynomial
\begin{align*}
    P_{\calC}(x_1,\dots,x_m) \coloneqq \sum_{C \in \calC} \prod_{e\in C} x_c.
\end{align*}
Note that $P_{\calC}$ is multilinear and has all coefficients equal to $1$.
Further, if $(\bb_1,\ldots,\bb_m) \sim \plant_{\calH,z}^\rho$, then by \Cref{fact:even cover in planted distribution}, $P_\calC(\bb_1,\ldots,\bb_m) = \sum_{C \in \calC} \prod_{e \in C} \bm{\eta}_e$.
Hence $\Ex P_\calC(\bb_1,\ldots,\bb_m) = \sum_{C \in \calC} \rho^{|C|}$.
By our choice of $\rho$ and since $|C| \le 2T$, this will be noticeably different between the $\Null$ and $\Planted$ cases
($0$ in the former and noticeably positive in the latter).
However, we do not know how to prove concentration of $P_\calC(\bb_1,\ldots,\bb_m)$ in the $\Planted$ case directly.
Instead, we apply the anticoncentration machinery developed in the previous section
to a random block-linear restriction of the polynomial $P_\calC$.
(We only need a simple property of the random partition,
namely, that the corresponding block-linear restriction is sufficiently large.)

Our algorithm is as follows:

\begin{algorithm}[H]
\caption{\textsc{Distinguish-Using-Even-Covers}$(\calC,\{b_e\}_{e\in E})$}\label{alg:distinguisher}
    \KwData{A $k$-uniform hypergraph $\calH = (V,H)$, a set of even covers $\calC$ of $\calH$, each of length at most $2T$, with $|\calC| = \lceil 10N^\eps \rceil$, and hyperedge signs $(b_e \in \{\pm1\})_{e \in H}$.}
    \For{$s=1,\ldots,S\coloneqq10 e^{2T}$}{
        Sample a random partition $\bm{\pi}$ of $H$ into $2T$ equally-sized parts\;
        Let $\calC'\subseteq \calC$ be the set of even covers shattered by $\bm{\pi}$\;
        If $|\calC'|\geq N^{\epsilon}\cdot 0.1^{T}$ then break the loop\;
    }
    If the loop never broke, output \textsc{Fail}\;
    If $\calC'$ contains more than $N^\epsilon \cdot 0.1^T$ even covers, restrict it to $N^\epsilon \cdot 0.1^T$ arbitrary ones\;
    Define $P_{\calC'}(x_1,\ldots,x_m) \coloneqq \sum_{C \in \calC'} \prod_{e \in C} x_c$\;
    Sample $\bm{\xi}_1,\ldots, \bm{\xi}_m\sim \Unif([0,1])$ i.i.d. and define $\bx_j \coloneqq \bm{\xi}_j b_j$ for $j \in [m]$\;
    If $P_{\calC'}(\bx_1,\ldots,\bx_m)\geq N^{0.6\epsilon}$ then output $\Planted$, otherwise output $\Null$\;
\end{algorithm}
The correctness theorem for this algorithm is as follows:
\begin{lemma}\label{lem:distinguisher correct}
    Let $\calH$ be a hypergraph, $\ell \in \N$, $N = \binom{n}\ell$, and $\rho,\delta > 0$.
    Let \[ 4\log(1/\delta) < T < \frac{0.4\eps\log N}{\log\bigl(10 (\Ckr)^{40}\,\rho^{-4}\bigr)} . \]
    Then \Cref{alg:distinguisher}, given as input a set $\calC$ of nonempty even covers of size $\le 2T$ in $\calH$,
    with $|\calC| \ge 10N^\eps$,
    outputs a guess \textsc{Null} or \textsc{Planted}, correct with probability at least $1-\delta$.
\end{lemma}
We will use the following lemmas in the proof of \Cref{lem:distinguisher correct}:

\begin{lemma}\label{lem:found good partition}
    The for loop in \Cref{alg:distinguisher} finds $\calC'$ that causes the loop to break with probability at least $0.99$.
\end{lemma}
\begin{proof}
    This follows from \Cref{lem:block}.
\end{proof}

From now on, we condition on the event that \Cref{alg:distinguisher} did not fail, and therefore
found a partition $\pi$ of $[m]$ and a large enough subset of even covers $\calC'\subseteq\calC$ that are $\pi$-block-linear (which we now also fix).
\Cref{lem:concentration about 0 null case} shows that in the $\Null$ case, the polynomial $P_{\calC'}$ concentrates about the value 0:

\begin{lemma}\label{lem:concentration about 0 null case}
    Fix $\epsilon,\delta > 0$ and assume that $4 \log(1/\delta) < T$.
    If $\bx_1,\dots,\bx_m$ are independent mean-zero random variables bounded between $-1$ and $1$, then:
    \begin{align*}
        \Pr_{\bx_1,\ldots,\bx_m} \sbra{ \abs{P_{\calC'}(\bx_1,\dots,\bx_m)} < N^{0.6\epsilon} } \ge 1-\delta.
    \end{align*}
\end{lemma}
\begin{proof}
    The polynomial $P_{\calC'}$ has constant term 0 since all of the even covers are nonempty. Since $P_{\calC'}$ is multilinear and the $\bx_i$ are independent with mean 0, we have that the expectation of $P_{\calC'}(\bx_1,\dots,\bx_m)$ is 0.
    Moreover, the second moment of $P_{\calC'}(\bx_1,\dots,\bx_m)$ is
    \begin{align*}
        \Ex_{\bx_1,\dots,\bx_m}\sbra{P_{\calC'}(\bx_1,\dots,\bx_m)^2} &= \Ex\sbra{\sum_{C,C' \in \calC'} \prod_{e\in C}\bx_e\prod_{e\in C'}\bx_e}\\
        &=\Ex \sbra{\sum_{C \in \calC'} \prod_{e\in C}\bx_e^2}+\Ex \sbra{\sum_{C \ne C'} \prod_{e\in C}\bx_e\prod_{e\in C'}\bx_e}.
    \end{align*}
    Now the second term on the RHS vanishes since $C$ and $C'$ are distinct and $\bx_e$'s are independent and mean-zero,
    and we upper-bound the first term by $N^\epsilon \cdot 0.1^T$ by boundedness of the $\bx_e$'s.
    Finally, by Chebyshev's inequality, we have that
    \begin{align*}
        \Pr\sbra{P_{\calC'}(\bx_1,\dots,\bx_m)\ge N^{0.6\epsilon}} &\leq \frac{(N^{\epsilon} \cdot 0.1^T)}{(N^{0.6\epsilon})^2} = N^{-0.2\eps} \cdot 0.1^T \le 0.1^T \le \delta^{4 \log(10)} \le \delta.\qedhere
    \end{align*}
\end{proof}

Next, we will use \Cref{lem:not concentrated about 0 in planted case} to show that the polynomial $P_{\calC'}$ corresponding to the $\calC'$ found by \Cref{alg:distinguisher} does not concentrate around 0 in the $\Planted$ case.
Let $\calD_{\text{plant}} \coloneqq\bigl(\tfrac12-\tfrac12\rho\bigr)\Unif([-1,0])+\bigl(\tfrac12+\tfrac12\rho\bigr)\Unif([0,1])$.
Note that $\Ex \calD_\plant = \tfrac12\rho$.

\begin{lemma}\label{lem:not concentrated about 0 in planted case}
    Fix $\epsilon, \delta > 0$, and assume $4 \log(1/\delta) < T < \frac{0.4\eps\log N}{\log(10(\Ckr)^{40}\,\rho^{-4})}$.
    If $\bx_1,\dots,\bx_m \overset{\text{iid}}{\sim} \calD_{\text{plant}}$ then:
    \[
        \Pr_{\bx_1,\ldots,\bx_m} \sbra{ \abs{P_{\calC'}(\bx_1,\dots,\bx_m)} \geq N^{0.6\epsilon}} \ge 1-\delta.
    \]
\end{lemma}
\begin{proof}
    The polynomial $P_{\calC'}$ is $\pi$-block-linear, where $\pi$ was the partition found by \Cref{alg:distinguisher}. Note that
    \[
    \Ex[P_{\calC'}(\bx_1,\ldots,\bx_m)] = \sum_{C \in \calC'} (\tfrac12\rho)^{|C|} \ge |\calC'| \cdot (\tfrac12\rho)^{2T} \ge N^\eps \cdot 0.1^T \cdot (\tfrac12\rho)^{2T}.
    \]
    By \Cref{prop:block linear anticoncentration}, we have
    \begin{align*}
        \Pr\sbra{\abs{P_{\calC'}(\bx_1,\dots,\bx_m)} < (\Ckr)^{-40T}(\tfrac12 \rho)^{2T} 0.1^T N^{\eps}} \leq 0.9^{2T} \le 0.9^{8\log(1/\delta)} = \delta^{8 \log(1/0.9)} \le \delta.
    \end{align*}
    Finally, we observe that
    \[
        N^{0.6\eps} < (\Ckr)^{-40T} (\tfrac12\rho)^{2T} 0.1^T N^\eps
    \]
    because dividing by $N^\eps$, reciprocating and taking the log of both sides, we get $T \log (\Ckr^{40} (\tfrac12 \rho)^{-2} 10) < 0.4\eps \log N$,
    which is exactly the assumption on $T$.
\end{proof}

Finally, we prove \Cref{lem:distinguisher correct}:

\begin{proof}[Proof of \Cref{lem:distinguisher correct}]
    Recall the definition of the distribution $\plant^\rho_{\calH,z}$:
    We have a \emph{fixed} planted assignment $z_1,\ldots,z_n \in \{\pm1\}$ and hyperedges $e \subseteq [n]$ of size $k$.
    For each hyperedge $e$, we sample $\bb_e = \bm{\eta}_e \prod_{v \in e} z_v$,
    where $(\bm{\eta}_e \in \{\pm1\})_{e \in H}$ are i.i.d.~with $\Ex \bm{\eta}_e = \rho$.

    Then \Cref{alg:distinguisher} samples i.i.d.\ $\bm{\xi}_1,\ldots,\bm{\xi}_m \sim \Unif([0,1])$, sets $\bx_j = \bm{\xi}_j b_j$, and evaluates
    \[
    P_{\calC'}(\bx_1,\ldots,\bx_m) = \sum_{C \in \calC'} \prod_{v \in C} \bx_j.
    \]
    Now defining $\by_j \coloneqq \bm{\eta}_j \bm{\xi}_j$, by \Cref{fact:even cover in planted distribution},
    for every even cover $C$ we have $\prod_{v \in C} \bx_j = \prod_{v \in C} \by_j$, and so
    \[
    P_{\calC'}(\bx_1,\ldots,\bx_m) = \sum_{C \in \calC'} \prod_{v \in C} \by_j = P_{\calC'}(\by_1,\ldots,\by_m).
    \]
    Finally, we examine the distribution of the random variables $\by_j$ in the $\Null$ and $\Planted$ cases.
    In both cases, since all $\bm{\xi}_j$'s and $\bm{\eta}_j$'s are independent, the $\by_j$'s are as well.
    In the \textsc{Null} case, $\bm{\eta}_j$ is uniformly $\pm 1$ and $\bm{\xi}_j$ is uniform in $[0,1]$,
    hence $\by_j$ is mean zero and bounded in $[-1,+1]$.
    In the \textsc{Planted} case,
    each $\bm{\eta}_j$ is supported on $\{\pm1\}$ with $\Ex \bm{\eta}_j = \rho$,
    and each $\bm{\xi}_j$ is again uniform in $[0,1]$.
    Hence $\by_j$ is independently distributed as $\calD_{\text{plant}}$.
    Thus, we can apply \Cref{lem:concentration about 0 null case,lem:not concentrated about 0 in planted case} to the two cases respectively
    and deduce the required conclusions.
\end{proof}
\subsection{Proof of main theorem}\label{sec:proof of main theorem}
Our main theorem, now with more parameters, is as follows:

\begin{theorem}[Main theorem]\label{thm:main formal}
    There exists a universal constant $1<C<\infty$\footnote{This is the constant from \Cref{lem:linear anticoncentration}.} such that the following holds.
    Let $n,k,m,\ell \in \N$ and $\rho,\epsilon,\delta > 0$,
    \[
        \bar{d}(n,m,k,\ell)\geq120 \cdot \epsilon \cdot  \left(10 \cdot C^{40} \cdot \left(\frac{2}{\rho}\right)^2\right)^{10/\epsilon} \cdot \log N.
    \]
    and $4 \log(1/\delta) \le T \coloneqq \frac{0.4\eps\log N}{\log(10(\Ckr)^{40}\,(\rho/2)^{-2})}$.
    Then, given a \emph{worst case $k$-uniform hypergraph} with $n$ variables and $m$ clauses,
    there is a classical algorithm running in $\tilde{O}(n^{(1/2+\epsilon)\ell+k+2} \log(1/\delta))$ time
    which distinguishes right-hand sides $\bb$ sampled from either the $\Null$ or $\rho$-$\Planted$ distributions w.p.\ $1-\delta$.
\end{theorem}
Our overall main theorem, \Cref{thm:main informal}, follows from this theorem by using the estimate for $\bar{d}$ in \Cref{rem:kikuchi params} and solving for $\ell$. Then we get:
\begin{align*}
    \ell &\geq \pbra{120\epsilon\pbra{40\cdot C^{40}\cdot \frac1{\rho^2}}^{10/\epsilon}}^{\frac2{k-2}} \pbra{\rho^2\binom{k}{k/2}}^{-\frac2{k-2}}(1/\delta)^{\frac2{k-2}}
\end{align*}
We set $ \epsilon = \frac{10 \log (1/\rho)}{\log k}$ and $\delta = N^{-\Omega(\frac{\log(1/\rho)}{\log k})} = n^{-\Omega(\frac{\ell \log(1/\rho)}{\log k})}$ in \Cref{thm:main formal} to recover \Cref{thm:main informal}.

In turn, \Cref{thm:main formal} follows immediately from \Cref{lem:find many walks,lem:distinguisher correct}:

\begin{proof}[Proof of \Cref{thm:main formal}]
    Noting that $N^{1/\log N} = 2$, we can invoke \Cref{lem:find many walks} so long as the average degree of the Kikuchi graph $K_\ell(\calH)$ is least \[
    300N^{4/T} T = 300 \cdot 2^{\pbra{ 4 \log(10(\Ckr)^{40}\,(\rho/2)^{-2})/0.4\eps }} \cdot \frac{0.4\eps\log N}{\log(10(\Ckr)^{40}\,(\rho/2)^{-2})}\]
    for sufficiently small $\rho > 0$.
    Given this, by \Cref{lem:find many walks}, \Cref{alg:find many walks} (with small constant failure parameter) will output
    $10N^\eps$ nontrivial closed walks in the Kikuchi graph of length $\le 2T$ such that
    for each walk $W$, $\oddcolors_\chi(W)$ is distinct.
    Correspondingly, by \Cref{fact:walks and even covers}, we get $10N^\eps$ distinct nonempty even covers of $\calH$,
    each of size at most $2T$.
    Finally, we apply \Cref{lem:distinguisher correct} to get a distinguisher based on these even covers. It remains to analyze the runtime of the various algorithms.

    \paragraph{Time to find even covers.}
    We represent vertices in the Kikuchi graph $K_\ell(\calH)$
    as length-$\ell$ sorted strings of indices in $[n]$.
    Thus, we can take the symmetric difference of two vertices $w_1$ and $w_2$ in time $\tilde{O}(\ell)$,
    and we can calculate the degree of a vertex $w$ or sample a random neighbor of $w$ (along with the edge-color, as an index in $[m]$, of the corresponding edge) in time $\tilde{O}(m \ell)$ (by enumerating all hyperedges in $H$).

    To run \Cref{alg:find one walk} given a start vertex $v_0$, we first
    sample $L$ walks from $v_0$ of length $T$ and check if each is good,
    which takes time $\tilde{O}(LTm\ell)$ by the previous paragraph.
    Then, we look for collisions among the good walks by
    inserting their endpoints into a balanced binary search tree, which takes time $\tilde{O}(L \ell \log L)$.
    If we find a collision, we can calculate $\oddcolors_\chi(W)$ for the resulting walk $W$,
    encoded as a sorted list of indices in $[m]$ of length $\le 2T$,
    in time $\tilde{O}(T \log T)$.

    Next, in \Cref{alg:find many walks},
    \Cref{alg:find one walk} is called with a start vertex sampled from the Kikuchi graph's stationary distribution.
    To sample efficiently, we first sample a uniformly random hyperedge $e \in H$,
    then sample a random subset $e' \subseteq e$ with $|e'| = k/2$,
    then finally sample a random subset $f \subseteq [n] \setminus (e \setminus e')$
    with $|f| = \ell-k/2$ and output $e' \cup f$.
    This takes time $\tilde{O}(\ell)$.
    If \Cref{alg:find one walk} succeeds in outputting a walk $W$,
    we then check if $\oddcolors_\chi(W)$ is nonempty (i.e., $W$ gives a nonempty even cover).
    If so, we check $\oddcolors_\chi(W)$ against a balanced binary search tree consisting of even covers we have already found.
    We do this $R$ times and so the checks take time $\tilde{O}(TR \log R)$.

    Altogether, using $N \le n^\ell$ and $m \le n^k$, we have $L = O(n^{\ell/2})$, $T = \tilde{O}(n)$, and $R = O(n^{\epsilon \ell} + \log(1/\delta))$.
    Hence the runtime for \Cref{alg:find one walk} is $\tilde{O}(n^{\ell/2+k+2})$,
    and so the runtime for \Cref{alg:find many walks} is $\tilde{O}((n^{(1/2+\epsilon)\ell+k+2}) \cdot \log(1/\delta)) $.
    
    \paragraph{Time to use even covers.}
    After \Cref{alg:find many walks} completes, we move on to \Cref{alg:distinguisher}.
    In each iteration of \Cref{alg:distinguisher}'s loop,
    we can sample an equipartition $\bm{\pi}$ of $H$ by sampling a real number in the unit interval for each $e \in H$
    and then sorting $H$ by these numbers,
    which takes $\tilde{O}(m)$ time.
    Then, we iterate through the even covers in $\calC$
    and check whether each is shattered by marking the partition's parts.
    Across all even covers, this takes time $O(|\calC| T)$.
    The loop in \Cref{alg:distinguisher} runs $S = O(e^{2T}) = O(2^{(1.6 \log e) \log N}) = O(N^{2.4})$ times.
    Finally, once we break the loop, the final step (of calculating $P_{\calC'}$) takes
    $\tilde{O}(|\calC'| T) = \tilde{O}(N^\epsilon \cdot 0.1^T \cdot T) = \tilde{O}(N^\epsilon)$ time.
\end{proof}

\section*{Acknowledgments}
We thank Yang Liu and Mehtaab Sawhney for helpful discussions on polynomial anticoncentration.
We thank Omar Alrabiah, Ryan Babbush, Venkat Guruswami, Pravesh Kothari, Peter Manohar, and Sidhanth Mohanty for general discussion.
We thank Matan Shtepel and Junzhao Yang for their contributions to early stages of this work.

\bibliographystyle{alpha}
\bibliography{references}

\end{document}